\documentclass[11pt]{article}

\usepackage{amsthm,amsmath,amssymb}
\usepackage{graphicx,color,epsfig}
\usepackage[margin=1in]{geometry}

\newtheorem{Thm}{Theorem}[section]

\newtheorem{Lem}[Thm]{Lemma}
\newtheorem{Pro}[Thm]{Proposition}

\newcommand{\N}{\mathbb{N}}
\newcommand{\R}{\mathbb{R}}

\DeclareMathAlphabet{\mathpzc}{OT1}{pzc}{m}{it}

\begin{document}

\title{Revisiting asymptotic periodicity in networks of degrade-and-fire oscillators}
\author{Bastien Fernandez}
\date{}   
\maketitle

\begin{center}
Laboratoire de Probabilit\'es, Statistique et Mod\'elisation\\
CNRS - Univ. Paris 7 Denis Diderot -  Sorbonne Univ.\\
75205 Paris CEDEX 13 France\\
fernandez@lpsm.paris
\end{center}

\begin{abstract}
Networks of degrade-and-fire oscillators are elementary models of populations of synthetic gene circuits with negative feedback, which show elaborate phenomenology while being amenable to mathematical analysis. In addition to thorough investigation in various examples of interaction graphs, previous studies have obtained conditions on interaction topology and strength that ensure that asymptotic behaviors are periodic (assuming that the so-called firing sequence is itself periodic and involves all nodes). Here, we revisit and extend these conditions and we analyse the dynamics in a case of unidirectional periodic chain. This example shows in particular that the updated conditions for asymptotic periodicity are optimal. Altogether, our results provide a novel instance of direct impact of the topology of interactions in the global dynamics of a collective system.
\end{abstract}

\leftline{\small\today.}
\bigskip

\centerline{\sl Dedicated to the memory of Valentin Afraimovich.}

\section{Networks of degrade-and-fire oscillators}
One important challenge in nonlinear dynamics is to evaluate the impact of interaction topology (and strength) on the functioning of networks of interacting units \cite{AD-GKMZ08,BCFV02}. How does the way the units are coupled affect the long term organization of a collective system? This question has received considerable attention in the theoretical literature \cite{BLMCH06} and is motivated by applications in many disciplines, especially in Physics, Biology and Social Sciences \cite{S01}. From a rigorous mathematical viewpoint, results are however scarce and limited \cite{CF05}. Studies have addressed either special cases such as weak-coupling regimes  or synchronisation \cite{B97,HI97} or specific systems such as pulse-coupled oscillators with excitatory coupling \cite{B95,MS90}. Hence, there is a critical need for additional insights, especially those that can provide the theory with complementary rigorous footing.

Recently, a simple dynamical model of populations of pulsatory oscillators with inhibitory coupling has been proposed, that shows non-trivial phenomenology while being  amenable to mathematical analysis. The model, a network of degrade-and-fire oscillators, results from the simplification of some system of delay-differential equations \cite{MHT14} inspired from experiments on colonies of synthetic gene circuits \cite{DM-PTH10,M-PDSTH11}. It can be summarized as follows, see e.g.\ \cite{HM-MIC01,SBB00} for more details about modelling of gene regulatory networks. Assuming a population of $N$ cells, let $x=(x_i)_{i=1}^N\in [0,1]^N$ be the vector of (normalized) expression levels of a certain gene in each cell. Furthermore, given a threshold parameter $\eta\in (0,1)$ and $W=(w_{ij})_{i,j=1}^N$ a stochastic non-negative matrix satisfying $w_{ii} > \eta$ for all $i$, let the vector $Wx=(Wx_i)_{i=1}^N$ where 
\[
W x_i = \sum_{j=1}^N w_{ij} x_j,\ \forall i\in \{1,\dots,N\},
\]
represent the expression levels of a related repressor agent. The matrix $W$ captures both interaction topology and its strength. It materializes the assumption that the repressive constituent is both directly proportional to gene expression level (self-repressor gene) and sufficiently small to diffuse through cell membranes (key experimental principle for inter-cellular coupling between genes \cite{DM-PTH10,M-PDSTH11}).

Now, the evolution of time-dependent expression levels $x(t)=(x_i(t))_{i=1}^N$ for $t\in\R^+$ is given by the following singular differential equation
\begin{align}\label{eq:diffEQ}
\begin{split}
\begin{array}{l l}
\dot{x}_i(t) = - \text{Sgn}(x_i(t)) & \text{if} \quad W x_i(t) > \eta, \\
\bigg\{ \begin{array}{l}x_i(t) = x_i(t-0) \\ x_i(t + 0) = 1 \end{array} & \text{if} \quad W x_i(t) \leq \eta.
\end{array} 
\end{split}
\end{align}
In other words, the dynamics in cell $i$ consists of two phases, depending on the value of $Wx_i(t)$.
\begin{itemize}
\item  If $W x_i(t) > \eta$, then the expression level $x_i(t)$ {\bf degrades} at constant speed $-1$, unless $x_i(t)=0$ (in which case, it remains steady). This phase may eventually yield $W x_i(t)\leq \eta$, depending on expression levels in influencing cells. 
\item If $W x_i(t) \leq \eta$, a {\bf firing} takes place and resets the expression level to the value 1. The assumption $w_{ii} > \eta$ ensures that $W x_i(t+0) > \eta$ immediately after firing and the reset genes return to the degrade phase for a positive-length time interval. Actually, this assumption allows one to prove that the dynamics is globally well-posed (and $W x_i(t) \geq \eta$ for all $t$), when assuming that all cells are initially in the degrading phase (Lemma 2.1 in \cite{BF16}).
\end{itemize}
In each cell, the expression level behavior thus consists of an infinite succession of degrading phases interrupted by instantaneous firings, unless the repressor level eventually remains above the threshold. Depending on $W$, some genes may indeed reach a vanishing stationary state. But if the external repressor contribution to cell $i$ is not too important (ie.\ $\sum_{j\neq i}w_{ij}<\eta$), then one can show that the corresponding gene never stops firing (Lemma 3.1 in \cite{BF16}); hence convergence to a stationary state is impossible.

In addition, trajectory asymptotic behaviors depend on firing patterns, viz.\ on the way the firings are distributed in the population. Depending again on $W$, repetitive firings may occur in some cells between two consecutive resets in other cells. However, if $W$ is doubly stochastic and $w_{ij}<\tfrac1{N}$ for all $j\neq i$, then every cell $j\neq i$ must fire between any two consecutive firings in cell $i$ alone. When this property holds for all cells, the sequence $\{i_k\}_{k\in\N}$ where $i_k$ denotes the reset cell label at firing $k$, is said to be {\bf exhaustive}. In other words, an exhaustive firing sequence is periodic and each of its length $N$ segment must be a permutation of $\{1,\dots ,N\}$.

The occurrence of exhaustive firing sequences has been identified as a nice and crucial feature of the degrade-and-fire dynamics. In fact, under certain conditions on $W$ (which are discussed below), it implies asymptotic periodic behavior of the expression levels themselves. The limit trajectory is a closed loop, which is unique (and explicitly computable) for any given firing word $\{i_k\}_{k=1}^N$. Usually, the loop is a {\bf minimal} periodic orbit of \eqref{eq:diffEQ} (ie.\ it returns to its initial location after every cell has fired once). Nonetheless, due to the singular nature of the dynamics, there are exceptional parameter values for which the loop itself is not invariant \cite{BF17}. 

Therefore, if $W$ is also such that every firing sequence must be eventually exhaustive, it suffices to examine the existence of the corresponding periodic orbits/loops to characterize the asymptotic dynamics. This approach yielded full description of the population dynamics for mean-field coupling $Wx_i=(1-\epsilon)x_i+\frac{\epsilon}{N}\sum_{j=1}^Nx_j$ for $\epsilon \in (0,1-\eta)$ (where exhaustivity should be understood up to cluster considerations, see \cite{FT11}) and also for some low dimensional examples \cite{BF16}.

The purpose of this paper is to revisit and improve the conditions on $W$ for periodic behavior under the exhaustive firing sequence assumption. In particular, we provide an alternative condition for uniqueness of the minimal periodic loop associated with a given exhaustive sequence (Section 2). More importantly, we improve previous conditions for asymptotic periodicity (Section 3) and show that these updated conditions are optimal, namely we provide a counter-example where asymptotic periodic fails when the conditions do not hold (Section 4). More precisely, we consider a periodic 3-cell chain in a parameter regime where every firing sequence must be eventually exhaustive and we show that when the firing pattern does not satisfy the improved conditions, no convergence to a single periodic loop takes place. Instead, a full family of non-minimally periodic trajectories exist which share the same firing sequence. 

Altogether, our results provide a novel instance of direct impact of the topology of interactions in the global dynamics of a collective system. We hope that the methods developed here can be applied to other systems with oscillatory dynamics. 

\section{Periodic orbit uniqueness}  
In this section, we shall be concerned with periodic orbit that are {\bf non-degenerate}, namely such that every expression level that vanishes at some time is reset at the next firing. In \cite{BF16}, we proved uniqueness of non-degenerate minimal periodic orbits associated with exhaustive firing sequences, under the assumptions that $W$ be irreducible and doubly-stochastic. Here, we show that the last assumption is not necessary and can be replaced by aperiodicity.
\begin{Pro}
Assume that $W$ is primitive. Then, given any exhaustive firing sequence, either no compatible non-degenerate minimal periodic orbit exists, or such a trajectory is unique.  
\label{UNIPER}
\end{Pro}
The matrix $W$ associated with mean-field coupling is evidently primitive (and also doubly-stochastic); hence periodic orbit uniqueness for mean-field coupling can be regarded as a special case of Proposition \ref{UNIPER}. Furthermore, while optimal conditions on $W$ for this Proposition's claim to hold are yet to be determined, limitations must be imposed because counter-examples exist for which the claim does not hold. Indeed, for $W=\text{Id}$ (ie.\ the diagonal matrix with all diagonal entries equal to 1), every configuration $(x_1,x_2,\dots,x_{N-1},1)$ with $\eta<x_1$ and $x_i<x_{i+1}$ for all $i$, is periodic with firing sequence defined by repeating the word $(1,2,\dots, N)$.

\begin{proof} 
The proof follows the same lines as the proof of Proposition 4.1 in \cite{BF16}. Given an arbitrary permutation $\pi$ of the cell indexes $\{1,\dots ,N\}$, we assume that a non-degenerate minimal periodic trajectory $t\mapsto x(t)$ exists with firing sequence obtained by repeating the word $(\pi_i)_{i=1}^N$. We prove that this trajectory must be unique, given $W$ and $\eta$. 

Let $(t_k)_{k=1}^N$ with $t_k<t_{k+1}$ be the first N firing times and let $R_\pi$ be the representation of $\pi$ on $\R^N$, ie.\ $R_\pi x_i=x_{\pi_i}$ for all $i$. Imposing periodicity after the $N$th firing, viz.\  $x(0)=x(t_N+0)$, implies that the initial coordinates $x(0)=x$ must be given by 
\begin{equation}
R_\pi x_i=1-t_N+t_i,\ \forall i\in\{1,\dots,N\}.
\label{eq:periodicity}
\end{equation}
We aim to show that the firing times $(t_k)_{k=1}^N$ must be unique. Their monotonicity implies $R_\pi x_i<R_\pi x_{i+1}$, and also $R_\pi x_1>0$ from the non-degeneracy assumption and $R_\pi x_N=1$. From \eqref{eq:periodicity}, the differences $R_\pi x_i-t_i=1-t_N$ do not depend $i$. To proceed, we separate the cases $1-t_N>0$ and $1-t_N\leq 0$.

Assume the first case $1-t_N>0$. From \eqref{eq:periodicity} and \eqref{eq:diffEQ}, one obtains the following gene expression levels immediately before the $i$th firing:
\[
R_\pi x_j(t_i)=\left\{\begin{array}{ccl}
1-t_i+t_j&\text{if}&j<i \, ,\\
1-t_N+t_j-t_i&\text{if}&j\geq i.
\end{array}\right.
\]
Therefore, the set of equations $R_\pi Wx_i(t_i)=\eta$ for $i\in\{1,\dots,N\}$ can be written in a condensed form as $R_\pi \Delta x=u$ where $\Delta=W-\text{Id}$ and
\[
u_i=1-\eta-t_Nv_i\quad\text{where}\quad v_i=\sum_{j=i}^Nw_{\pi_i\pi_j}\ \forall i\in\{1,\dots,N\}.
\]
By Perron-Frobenius Theorem, the normalized left eigenvector $(e_W)^T=(e_W)^TW$, associated with the eigenvalue 1, is unique. Moreover, the definition of $\Delta$ obviously implies 
\[
\Delta(\R^N)\subset \Sigma_W=\left\{x\in\R^N\ :\ (e_W)^T.x=0\right\}.
\]
From the firing time equation, it follows that $R_\pi^{-1}u\in \Sigma_W$ (NB: $R_\pi$ is a permutation matrix, hence $R_\pi^{-1}$ is well defined) and hence $(e_W)^T.R_\pi^{-1}u=0$. The expression of $u$, the property $(e_W)^T.(1)_{i=1}^N=1$ (Perron-Frobenius) and the $\ell_1$-normalization of $e_W$ then yield the expression (NB: we must have $(e_W)^T.R_\pi^{-1}v>0$ so that $t_N$ is well-defined.)
\[
t_N=\frac{N(1-\eta)}{(e_W)^T.R_\pi^{-1}v},
\]
Therefore, $u$ only depends on $\pi,W$ and $\eta$ and the condition $1-t_N>0$ is equivalent to $\frac{N(1-\eta)}{(e_W)^T.R_\pi^{-1}v}<1$.

It remains to solve the firing time equation $R_\pi \Delta x=u$. Writing $x=x_\parallel+x_\perp$, where $x_\parallel=c(e_W)^T$ for some $c\in\R$ and $x_\perp\in\Sigma_W$, the equation $R_\pi \Delta x=u$ becomes $\Delta x_\perp=R_\pi^{-1}u$. Moreover, given any $x_\perp$, the constant $c$ is determined by using the normalization $x_{\pi_N}=1$, {\sl i.e.}\ $c(e_W)_N=1-(x_\perp)_{\pi_N}$.

Now, Perron-Frobenius also implies that $\|W|_{\Sigma_W}^k\|_{\ell_1}<1$ provided that $k$ is sufficiently large ($W$ is primitive). Hence $W-\text{Id}$ is invertible on $\Sigma_W$. Using that $R_\pi^{-1}u\in \Sigma_W$, this implies that $\Delta x_\perp=R_\pi^{-1}u$ has a unique solution in $\Sigma_W$. Uniqueness is proved in the case 
$1-t_N>0$. 

In the case where $1-t_N\leq 0$, we have 
\[
R_\pi x_j(t_i)=\left\{\begin{array}{ccl}
1-t_i+t_j&\text{if}&j<i \, , \\
0&\text{if}&j=i \, ,\\
1-t_N+t_j-t_i&\text{if}&j> i \, ,
\end{array}\right.
\]
and the firing time equation now reads $\Delta x=R_\pi^{-1}u'$ where 
\[
u'_i=1-\eta-v'_i-t_Nv''_i\quad\text{where}\quad v'_i=w_{\pi_i\pi_i}\ \text{and}\ v''_i=v_i-w_{\pi_i\pi_i}\ \forall i\in\{1,\cdots,N\}.
\]
Using as before the condition $R_\pi^{-1}u'\in\Sigma_W$, we obtain that the last firing time $t_N$ must be given by
\[
t_N=\frac{N(1-\eta)-(e_W)^T.R_\pi^{-1}v'}{(e_W)^T.R_\pi^{-1}v''}.
\]
Notice that the condition $1-t_N\leq 0$ is exactly complementary to the previous one $\frac{N(1-\eta)}{(e_W)^T.R_\pi^{-1}v}<1$. Moreover, uniqueness follows as in the previous case by solving $\Delta x_{\perp}=R_\pi^{-1}u'$. Therefore, there is at most one solution in all cases, when the parameters $\pi,W$ and $\eta$ are given. 
\end{proof}

In addition to relaxing the double-stochasticity condition, one advantage of the assumption in Proposition \ref{UNIPER} is that it allows one to immediately address trajectories with silent genes. For simplicity, we assume in the next statement that only one gene, say $N$ w.l.o.g., is silent. The statement easily extends to cases where several genes never fire.  
\begin{Pro}
Assume that $W$ is primitive and that node $N$ is such that no node $i\neq N$ only receives input from $N$, {\sl viz}\ $w_{iN}<1$ for all $i=\{1,\dots ,N-1\}$.  
Then, given any permutation $\pi^{(N-1)}$ of $\{1,\dots,N-1\}$, either no compatible non-degenerate minimal periodic orbit with $x_N(t)=0$ for all $t\in\R^+$ exists, or such a trajectory is unique.  
\end{Pro}
\begin{proof}
In this case, the initial coordinate expression writes
\[
R_{\pi^{(N-1)}} x_i=1-t_{N-1}+t_i,\ \forall i\in\{1,\dots,N-1\},
\]
and, again, the cases $1-t_{N-1}>0$ and $1-t_{N-1}\leq 0$ have to be considered separately. In the first case, the firing time equations $W_{(N-1)}x_{\pi^{(N-1)}_i}(t_i)=\eta$ for $i\in\{1,\dots,N-1\}$ can be written as $R_{\pi^{(N-1)}}\Delta_{N-1} x=u$ where $\Delta_{N-1}=W_{N-1}-\text{Id}|_{\R^{N-1}}$ and 
\[
u_i=1-\frac{\eta+t_{N-1}v_i}{1-\omega_{\pi^{(N-1)}_iN}}\quad\text{with}\quad v_i=\sum_{j=i}^{N-1}w_{\pi^{(N-1)}_i\pi^{(N-1)}_j}\ \forall i\in\{1,\dots,N-1\}.
\]
By assumption, the normalized truncated matrix $W_{N-1}=\left(\frac{\omega_{ij}}{1-\omega_{iN}}\right)_{i,j=1}^{N-1}$ must be primitive. Hence, the same arguments as before apply and it follows that these equations have a unique solution, given the parameters. The reasoning is similar in the complementary cases $1-t_{N-1}\leq 0$. The details are left to the reader. 
\end{proof}

\section{Asymptotic periodicity}
Theorem 5.1 in \cite{BF16} combined the double-stochasticity assumption with the existence of a length 2 loop in the graph with adjacency matrix $W$ to ensure asymptotic periodicity (with minimal period) for every trajectory whose firing sequence is exhaustive. However, the existence of a length 2 loop appears to be a restrictive condition, which can only be satisfied for a very limited set of matrices. Accordingly, this section aims to replace this assumption by a more general one, so that it includes more  interaction matrices $W$. This broader assumption turns out to be firing pattern dependent and the example in the next Section justifies such dependence as a necessary condition. 

\begin{Pro}
Assume that $W$ is irreducible and doubly-stochastic and consider a permutation $\pi$ of $\{1,\dots,N\}$ that is compatible with a non-trivial loop in the graph with adjacency matrix $W$ (ie.\ if $i$ precedes $j\neq i$ in the loop, then so does it in $(\pi_k)_{k=1}^N$). Then, for any trajectory $t\mapsto x(t)$ with periodic firing sequence obtained by repeating $(\pi_k)_{k=1}^N$, we have
\[
\lim_{k\to+\infty} x(t_{kN}+0)=R_\pi x,
\]
where $t_\ell$ ($\ell\in\N$) is the trajectory $\ell$th firing time and $R_\pi x$ is the initial condition of the minimal periodic trajectory associated with $\pi$.
\label{ASYMPTHM}
\end{Pro}
We are only concerned with convergence to periodic orbit as a whole and not to its individual periodic points. Therefore, it suffices that the compatibility assumption above holds for an appropriate cyclic permutation of $\pi$ (or of the loop). In particular, when the graph has a length 2 loop (ie.\ $i\to j\to i$), compatibility holds for every exhaustive firing sequence; in other words, Theorem 5.1 in \cite{BF16} is a special case of Proposition \ref{ASYMPTHM} here. 

The proof of Proposition \ref{ASYMPTHM} follows the same lines as the proof of Theorem 5.1 in \cite{BF16} and consists in showing that a sufficiently large iterate of the return map after $N$ firing is a contraction in $\Sigma:=\{x\in \R^N\ :\ \sum_i x_i=0\}$, for the $\ell^1$-norm. Here we only provide details of the part not already covered in \cite{BF16}. 

Given $i\in \{1,\dots,N\}$, let $L_i$ be the matrix defined by (NB: $W$ irreducible implies $w_{ii}<1$ for all $i$)
\[
(L_{i})_{jk}=\frac{w_{ji}}{1-w_{ii}}\delta_{ik} \, , \ \forall k\neq j\in\{1,\cdots ,N\}\quad\text{and}\quad 
(L_{i})_{jj}=1-\delta_{ij} \, ,\ \forall j\in\{1,\cdots ,N\} \, ,
\]
and, given a permutation $\pi$, consider the product $L_\pi$ defined by
\[
L_\pi=L_{\pi_N}\dots L_{\pi_2}L_{\pi_1}.
\]
The desired contraction is a consequence of the following statement and of the fact that if $M$ is a scrambling matrix, then the transpose $M^T$ contracts in $\Sigma$ for the $\ell^1$-norm \cite{S79}.
\begin{Lem}
Assume that $W$ is irreducible and doubly-stochastic and consider a permutation $\pi$ of $\{1,\dots,N\}$ that is compatible with a non-trivial loop in the graph associated with $W$. There exists $k\in\N$ such that the row-stochastic matrix $(L_\pi^k)^T$ is scrambling.
\end{Lem} 
\begin{proof}
We are going to prove the existence ok $k\in\N$ such that 
\[
(L_\pi^k)_{\pi_1j}>0,\ \forall j\in\{1,\dots,N\},
\]
which is a special case of scrambling property for $(L_\pi^k)^T$. 

By irreducibility of $W$, given an arbitrary $j\in\{1,\dots,N\}$, let $(j_k)_{k=1}^K$ ($K\in\{1,\dots, N-1\}$) be the shortest word (of length at least 2) such that 
\[
w_{\pi_1j_K}w_{j_Kj_{K-1}}\dots w_{j_2j_1}w_{j_1j}>0.
\]
Let $k_0\in\{1,\dots,N\}$ be such that $\pi_{k_0}=j$. The definition of $L_\pi$ and $1-w_{ii}<1$ imply
\[
(L_{\pi_{k_0}}\dots L_{\pi_1})_{j_1j}\geq w_{j_1j}.
\]
Let now $k_1$ be such that $\pi_{k_1}=j_1$ and consider separately the cases $k_0<k_1$ and $k_1<k_0$. In the first case, we have 
\[
(L_{\pi_{k_1}}\dots L_{\pi_{k_0+1}})_{j_2j_1}=(L_{\pi_{k_1}})_{j_2j_1}\geq w_{j_2j_1},
\]
and then 
\[
(L_{\pi_{k_1}}\dots L_{\pi_1})_{j_2j}\geq w_{j_2j_1}w_{j_1j}.
\]
In the second case, we have 
\[
(L_{\pi_N}\dots L_{\pi_{k_0+1}})_{j_2j_1}=0,
\]
so no positive estimate holds for $(L_\pi)_{j_2j}$. However, we certainly have 
\[
(L_\pi)_{j_1j}\geq w_{j_1j}\quad\text{and}\quad
(L_{\pi_{k_1},}\dots L_{\pi_1})_{j_2j_1}=(L_{\pi_{k_1}})_{j_2j_1}\geq w_{j_2j_1},
\]
hence
\[
(L_{\pi_{k_1},}\dots L_{\pi_1}L_\pi)_{j_3j_1}\geq w_{j_2j_1}w_{j_1j} \, .
\]
By repeating this process, we obtain that there exists $n\in\{0,\dots, K\}$ for every $(j_k)_{k=1}^K$ such that 
\[
(L_{\pi_{k_K}}\dots L_{\pi_1}L_\pi^n)_{\pi_1j}\geq w_{\pi_1j_K}w_{j_Kj_{K-1}}\dots w_{j_2j_1}w_{j_1j}>0.
\]
In particular, for $j=\pi_1$, the assumption that $\pi$ is compatible with a non-trivial loop in the graph of $W$ implies
\[
(L_\pi)_{\pi_1\pi_1}=w_{\pi_1j_K}\dots w_{j_1\pi_1}>0.
\]
When $j\neq \pi_1$, we have $\pi_1\not\in\{\pi_{k_K+1},\dots, \pi_N\}$ and hence
\[
(L_{\pi_{N}}\dots L_{\pi_{k_K+1}})_{\pi_1\pi_1}=1\ \Longrightarrow (L_\pi^{n+1})_{\pi_1j}>0.
\]
Letting $k=\max_{\{j_k\}_{k=1}^K} n$, we can multiply $L_\pi^{n+1}$ by $L_\pi^{k-(n+1)}$ to obtain the desired estimate $(L_\pi^{k})_{\pi_1j}>0$. \end{proof} 

\section{Unidirectional 3-cell system}
In order to illustrate previous results, we consider here an example of interaction matrix $W$ that only satisfies the compatibility assumption of Proposition \ref{ASYMPTHM} for an appropriate cyclic permutation of some $\pi$ but not all (NB: in previous examples \cite{BF16,FT11}, the assumptions held for appropriate cyclic permutation of every $\pi$). We proceed to an exhaustive analysis of the dynamics that aims to determine all possible asymptotic behaviors depending on parameters. In particular, we show that when the compatibility assumption in Proposition \ref{ASYMPTHM} does not hold, the conclusion may fail as well; hence necessity. 

In this example, $N=3$ and the matrix $W$ is defined by
\[
W=\left(\begin{array}{ccc}
1-\epsilon&0&\epsilon\\
\epsilon&1-\epsilon&0\\
0&\epsilon&1-\epsilon
\end{array}\right),
\]
where $\epsilon\in (0,1-\eta)$ (so that the dynamics \eqref{eq:diffEQ} is well defined for all initial conditions $x$ such that $\min_i Wx_i>\eta$.). Notice that the dynamics commutes with cyclic permutations of indexes. 

\subsection{Transition graph}
The analysis first task is to determine possible transitions between the states immediately after firing, {\sl viz.}\ when one coordinate $x_i=1$. Thanks to the cyclic permutation symmetry, we can assume w.l.o.g.\ that $x_3=1$ and we are going to determine all possible states after the first firing, when starting from $x=(x_1,x_2,1)$ (such that $\min_i Wx_i>\eta$). The analysis investigates the two following cases separately
\begin{itemize}
\item[$\bullet$] Either $\min_i Wx_i-\eta< \min\{x_1,x_2\}$. Then $i_m=\arg\min_i Wx_i$ is the first cell to fire. It fires from above 0 and the first firing time is given by $t_f=Wx_{i_m}-\eta$.
\item[$\bullet$] Or $\min\{x_1,x_2\}\leq\min_i Wx_i-\eta$. Then, one needs to examine the expression of $Wx_i(t)$ for $t\geq \min\{x_1,x_2\}$ in order to compute $\min_i t_i$ where $t_i$ is defined by $Wx_i(t_i)=\eta$. Then $i_m=\arg\min t_i$ is the first cell to fire. It fires from 0 with firing time $t_{i_m}$ 
\end{itemize} 
For the sake of simplicity, we shall focus on trajectories for which all coordinates are always distinct, meaning that no two cells ever fire together. Indeed we shall see at the end of this Section that the complementary case deals with exceptional initial conditions (or parameters); so it is not relevant for modelling purpose.

For convenience, we denote by $S_{12}$ (resp.\ $S_{21}$) any configuration $(x_1,x_2,1)$ with $x_1<x_2$ (resp.\ $x_2<x_1$). The states $S_{23},S_{32},S_{13}$ and $S_{31}$ are defined similarly. The analysis of possible transition between these states yields two distinct graphs, one for $\epsilon\leq\tfrac12$, the other for $\epsilon> \tfrac12$, that are given on Figure \ref{TRANSGRAPHS}.
\begin{figure}[ht]
\begin{center}
\includegraphics*[width=60mm]{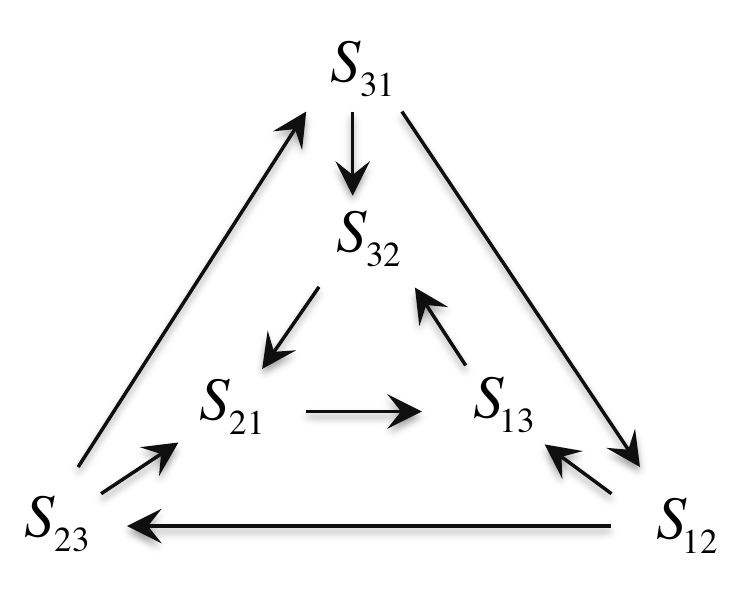}
\hspace{1cm}
\includegraphics*[width=60mm]{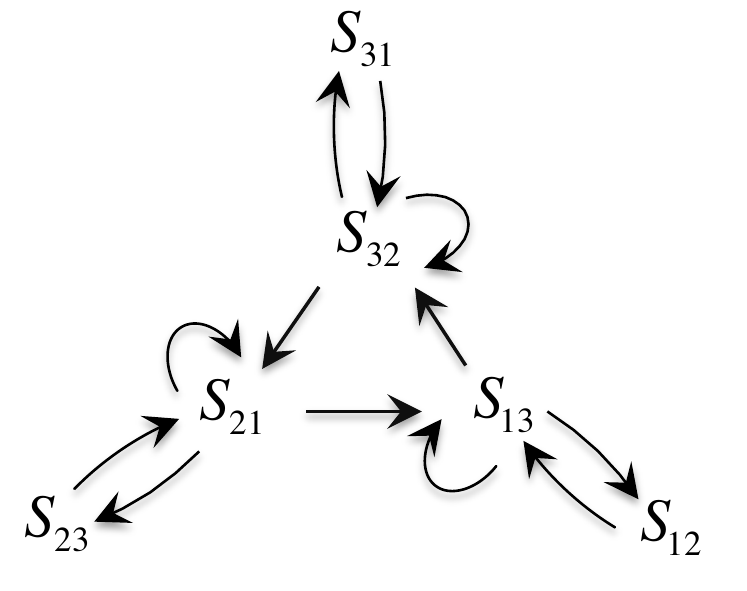}
\end{center}
\caption{Transition graphs between states immediately after firing in the unidirectional 3 cell system. Left: $\epsilon\leq\tfrac12$ Right: $\epsilon> \tfrac12$.}
\label{TRANSGRAPHS}
\end{figure}

\noindent
{\bf Analysis of possible transitions.} In order to determine $\min_i Wx_i$, we observe that
\[
Wx_1=(1-\epsilon)x_1+\epsilon,\ Wx_2=(1-\epsilon)x_2+\epsilon x_1, Wx_3=1-\epsilon +\epsilon x_2.
\]
Then, straightforward calculations yield
\[
\min_i Wx_i=\left\{\begin{array}{cl}
Wx_1&\text{if}\ \epsilon\leq \tfrac12\ \text{and}\ \frac{1-2\epsilon}{1-\epsilon}x_1+\frac{\epsilon}{1-\epsilon}<x_2\\
Wx_2&\text{if}\ \left\{\begin{array}{l}
\epsilon\leq \tfrac12\ \text{and}\ x_2< \frac{1-2\epsilon}{1-\epsilon}x_1+\frac{\epsilon}{1-\epsilon}\\
\text{or}\ \tfrac12<\epsilon\ \text{and}\ \frac{\epsilon}{2\epsilon-1}x_1-\frac{1-\epsilon}{2\epsilon-1}<x_2 \end{array}\right.\\
Wx_3&\text{if}\ \tfrac12<\epsilon\ \text{and}\ x_2< \frac{\epsilon}{2\epsilon-1}x_1-\frac{1-\epsilon}{2\epsilon-1}
\end{array}\right.
\]
(Of note, equality occurrence in the right inequalities here corresponds to $Wx_i=Wx_j$ for some $i\neq j$; hence to potential simultaneous firing in distinct cells, which we are not considering at this stage.) We study possible transitions in each case separately.
\begin{itemize}
\item[$\bullet$] $\epsilon\leq \frac12$ and $\frac{1-2\epsilon}{1-\epsilon}x_1+\frac{\epsilon}{1-\epsilon}<x_2$ ($\Longleftrightarrow Wx_1< Wx_2$), which implies in particular being in state $S_{12}$. There are two subcases:
\begin{itemize}
\item[-] Either $Wx_1-\eta\leq x_1=\min\{x_1,x_2\}$, then we have $S_{12}\to S_{23}$.
\item[-] Or $x_1<Wx_1-\eta\Longleftrightarrow x_1<1-\frac{\eta}{\epsilon}$ (which requires $\eta<\epsilon$). Then for $t\in [x_1,x_2]$, we have
\begin{equation}
Wx_1(t)=\epsilon (1-t),\ Wx_2(t)=(1-\epsilon)(x_2-t), Wx_3(t)=1-\epsilon +\epsilon x_2 -t.
\label{REPRF}
\end{equation}
Notice that $W_1(t)<W_3(t)$ for all $t$; thus cells $1$ and $2$ only can fire in this time interval (and one of them must fire before $Wx_2(t)$ reaches 0). Moreover, solving the equations $Wx_i(t)=\eta$ for $i=1,2$ shows that cell 1 fires iff 
\begin{equation}
1-\frac{\eta(1-2\epsilon)}{\epsilon(1-\epsilon)}\leq x_2.
\label{1FIRES}
\end{equation}
(NB: equality here corresponds to $1$ and $2$ simultaneously firing.) This inequality can perfectly hold when $\epsilon$ is close enough to $\tfrac12$; hence the transition $S_{12}\to S_{23}$ in this case. Otherwise, if $2$ fires (alone), we have $S_{12}\to S_{13}$.
\end{itemize}
\item[$\bullet$] $\epsilon\leq \frac12$ and $x_2< \frac{1-2\epsilon}{1-\epsilon}x_1+\frac{\epsilon}{1-\epsilon}$, which may occur both for $S_{12}$ and $S_{21}$.
\begin{itemize}
\item[-] Either $Wx_2-\eta\leq \min\{x_1,x_2\}$, which can perfectly happen for either state, depending on parameters (we have $Wx_2-\eta<x_1\Longleftrightarrow x_2<x_1+\frac{\eta}{1-\epsilon}$ and $Wx_2-\eta<x_2\Longleftrightarrow x_1-\frac{\eta}{\epsilon}<x_2$); hence the two transitions $S_{12}\to S_{13}$ and $S_{21}\to S_{13}$. 
\item[-] Or $x_1<Wx_2-\eta$ (which can only happen for $S_{12}$), then for $t\in [x_1,x_2]$, the repressor fields are given by \eqref{REPRF}; hence the transitions $S_{12}\to S_{13}$ and $S_{12}\to S_{23}$ as before.
\item[-] Or $x_2<Wx_2-\eta$ (which can only happen for $S_{21}$), then for $t\in [x_2,x_1]$, the repressor fields are given by
\[
Wx_1(t)=(1-\epsilon)x_1+\epsilon-t,\ Wx_2(t)=\epsilon(x_1-t), Wx_3(t)=(1-\epsilon)(1-t).
\]
Similarly to as before, we have $W_2(t)<W_3(t)$ here, so cells $1$ and $2$ only can fire. Letting $t_i$ be such that $Wx_i(t_i)=\eta$, we have 
\[
t_2=x_1-\frac{\eta}{\epsilon}<1-\frac{\eta}{1-\epsilon}=t_1\ \text{iff}\ x_1<1+\frac{\eta(1-2\epsilon)}{\epsilon(1-\epsilon)},
\]
hence the transition $S_{21}\to S_{13}$ in this case. 
\end{itemize}
\item[$\bullet$] $\tfrac12<\epsilon$ and $\frac{\epsilon}{2\epsilon-1}x_1-\frac{1-\epsilon}{2\epsilon-1}<x_2$, which may occur both for $S_{12}$ and $S_{21}$. The analysis is similar to the case $\epsilon\leq \frac12$ and $x_2< \frac{1-2\epsilon}{1-\epsilon}x_1+\frac{\epsilon}{1-\epsilon}$, except that when $x_1<Wx_2-\eta$, we must have $S_{12}\to S_{13}$ because \eqref{1FIRES} cannot hold for $\epsilon>\tfrac12$.
\item[$\bullet$] $\tfrac12<\epsilon$ and $x_2<\frac{\epsilon}{2\epsilon-1}x_1-\frac{1-\epsilon}{2\epsilon-1}$, which implies being in state $S_{21}$.
\begin{itemize}
\item[-] Either $Wx_3-\eta\leq x_2=\min\{x_1,x_2\}$ and then $S_{21}\to S_{21}$. 
\item[-] Or $x_2<Wx_3-\eta\Longleftrightarrow x_2<1-\frac{\eta}{1-\epsilon}$. Then one can repeat the analysis in the case $x_2<Wx_2-\eta$, to conclude that possible transitions are $S_{21}\to S_{13}$ and also $S_{21}\to S_{23}$, since we may now have $x_1>1+\frac{\eta(1-2\epsilon)}{\epsilon(1-\epsilon)}$.
\end{itemize}
\end{itemize}
In each case $\epsilon\leq \tfrac12$ and $\epsilon>\tfrac12$, the transition graphs of Figure \ref{TRANSGRAPHS} then follow by applying cyclic permutations. 

Furthermore, the analysis for $\epsilon\leq \tfrac12$ shows that the only way cells 1 and 2 fire together is when $Wx_1=Wx_2$ or when equality holds in \eqref{1FIRES}, which happens only two  specific segments of the square $[0,1]^2$ (viz. $x_2=\frac{1-2\epsilon}{1-\epsilon}x_1+\frac{\epsilon}{1-\epsilon}$ and $x_2=1-\frac{\eta(1-2\epsilon)}{\epsilon(1-\epsilon)}$). Then, cell 3 will necessary be the next cell to fire, and the reset state after second firing writes $(x,x,1)$ with $x<1$. Then either cell 1 or cell 2 fires alone at subsequent firing. Moreover, $x$ lies in some interval (only when $Wx_1=Wx_2$) or takes a single value. So possible states after the third firing lie inside one segment or a single point in the square $(x_1,x_3)\in [0,1]^2$ or $(x_2,x_3)\in [0,1]^2$. By applying cyclic permutation to the previous consideration, we conclude that subsequent simultaneous firing can only happen for a subset of the original segments; hence they consist of exceptional events in phase space or parameter space. Similar considerations apply for $\epsilon>\tfrac12$.

\subsection{Asymptotic dynamics for $\epsilon\leq\tfrac12$}
The left transition graph in Fig.\ \ref{TRANSGRAPHS} indicates that, for $\epsilon\leq\tfrac12$, the firing sequence of any trajectory (for which all coordinates remain distinct) must be periodic or eventually periodic and must satisfy the following alternative:
\begin{itemize}
\item[$\bullet$] Either it repeats (some cyclic permutation of) the word $(1,2,3)$. In this case, the trajectory cycles the outer loop $(S_{12},S_{23},S_{31})$ in Fig.\ \ref{TRANSGRAPHS}. By Propositions \ref{UNIPER} and \ref{ASYMPTHM}, it must asymptotically approach the unique minimal periodic orbit associated with $(1,2,3)$. Explicit calculations (below) show that this behavior takes place for an open set of initial conditions when $\epsilon<\tfrac13$, and does not occur for any trajectory when $\epsilon>\tfrac13$.
\item[$\bullet$] Or it (eventually) repeats (some cyclic permutation of) the word $(2,1,3)$ and the trajectory (eventually) cycles the inner loop $(S_{21},S_{13},S_{32})$. While uniqueness holds for the minimal periodic orbit associated with $(2,1,3)$ (Proposition \ref{UNIPER}), this word is not compatible with any loop in the graph of $W$ and Proposition \ref{ASYMPTHM} does not apply. In fact, while for $\epsilon<\eta$, every trajectory that cycles the inner loop must asymptotically approach this minimal periodic orbit, there exists $\epsilon_\eta\in (3\eta,\tfrac12]$ for every $\eta\in (0,\tfrac19)$ and for every $\epsilon\in (3\eta,\epsilon_\eta)$, there exists an open subset $U_{\epsilon,\eta}$ of the unit square such that, for every $(x_1,x_2)\in U_{\epsilon,\eta}$, the trajectory issued from $(x_1,x_2,1)$  is given by 
\[
(x_1,x_2,1)\to (\frac{\eta}{\epsilon},1,1-x_1+\frac{\eta}{\epsilon})\to (1,x_1,\frac{\eta}{\epsilon})\to (1-x_1+\frac{\eta}{\epsilon},\frac{\eta}{\epsilon},1)\to (\frac{\eta}{\epsilon},1,x_1)\to\dots.
\]
In particular, this trajectory has periodic firing sequence with repeated word $(2,1,3)$ but when $x_1\neq 1-x_1+\frac{\eta}{\epsilon}$, it is eventually non-minimally periodic because returns to reset states only occur after every cell has fired twice. 
\end{itemize}

\noindent
{\bf Analysis of periodic behaviors.} The first case of the above alternative is associated with the existence of the periodic orbit with firing pattern $(1,2,3)$. Thanks to commutation with cyclic permutations, these orbit states must have equi-distributed coordinates; in particular, the state $S_{12}$ must write $(1-2t_f,1-t_f,1)$ where $t_f$ is its firing time. Considerations in the previous Section imply that this periodic orbit exists iff $\epsilon<\tfrac13$ and certain conditions on the parameter $\eta$ hold. More precisely, 
\begin{itemize}
\item[$\bullet$] the orbit exists and fires from above 0 iff $\tfrac{2\epsilon}3<\eta$. (NB: The firing time $t_f$ is given by $t_f=\frac{1-\eta}{3-2\epsilon}$.)
\item[$\bullet$]the orbit exists and fires from 0 iff $\frac{\epsilon(1-\epsilon)}{2-3\epsilon}<\eta\leq \tfrac{2\epsilon}3$. (NB: The firing time $t_f$ is given by $t_f=1-\frac{\eta}{\epsilon}$.)
\end{itemize}

As existence of the family of periodic orbits in the alternative second case is concerned, we first want that the first firing when starting from $(x_1,x_2,1)$ occurs from 0 in cell 2, ie.\
\[
x_2<\min\{x_1,Wx_1-\eta,Wx_2-\eta\}.
\] 
Direct calculations show that these conditions are equivalent to 
\[
\frac{\eta}{\epsilon}<x_1\ \text{and}\ x_2<\min\{(1-\epsilon)x_1-\eta,x_1-\frac{\eta}{\epsilon}\}.
\]
Similar conditions are also required for the reset state $(\frac{\eta}{\epsilon},1,1-x_1+\frac{\eta}{\epsilon})$, namely we impose firing from 0 at cell 1. These conditions yield the following additional restriction
\[
x_1<1-\max\{\frac{2\eta}{1-\epsilon},\frac{\eta}{\epsilon}\}.
\]
Then, for the subsequent reset state $(1,x_1,\frac{\eta}{\epsilon})$, we have
\[
\max\{\frac{2\eta}{\epsilon},\frac{(1+\epsilon)\eta}{\epsilon(1-\epsilon)}\}<x_1,
\]
and the state after the next firing is $(1-x_1+\frac{\eta}{\epsilon},\frac{\eta}{\epsilon},1)$. Together with cyclic permutations, the inequalities here imply the existence of a non-minimal periodic orbit existence for every $(x_1,x_2)$ in some set $U_{\epsilon,\eta}$ provided that
\[
\max\{\frac{2\eta}{\epsilon},\frac{(1+\epsilon)\eta}{\epsilon(1-\epsilon)}\}<1-\max\{\frac{2\eta}{1-\epsilon},\frac{\eta}{\epsilon}\}.
\]
Straightforward calculations show that these inequalities are equivalent to $\eta\in (0,\tfrac19)$ and $3\eta<\epsilon<\frac{1-3\eta+\sqrt{1-10\eta+9\eta^2}}2$.

Finally, for $\epsilon<\eta$, every firing in a trajectory must occur from above 0. In this case, Remark 5.3 in \cite{BF16} implies that the minimal periodic orbit associated with $(2,1,3)$ attracts every trajectory that shares the same firing sequence.

\section{Concluding Remarks}
In this paper, we have considered the dynamics of networks of simple degrade-and-fire oscillators. In complement to previous results, sufficient conditions on the interaction matrix $W$ for asymptotic periodicity in presence of exhaustive firings have been obtained. Moreover, a counter-example has been given, which shows that the conditions are optimal. 

When combined with the conditions that exclude non-exhaustive firings (mentioned in the introduction) and other considerations (such as immersion in symmetry subspace and other gene silencing), these results provide a comprehensive description of the degrade-and-fire dynamics for a large panel of scenarios. 

Furthermore, while non-exhaustive firings cannot a priori be excluded in full generality, we believe that such firings are rare, if not exceptional in the asymptotic dynamics. For instance, one can check that a (minimal) periodic orbit with the simplest non-exhaustive firing sequence of the right transition graph in Fig.\ \ref{TRANSGRAPHS} cannot  exist (see analysis below). We believe that similar impossibility prevails for any other periodic trajectory with non-exhaustive firing in this graph. To a larger extent, to investigate the existence of recurrent behavior with non-exhaustive firing sequence in full generality will be the subject of future studies. 
\bigskip

\noindent
{\bf Periodic orbit with non-exhaustive firing sequence in the unidirectional 3-cell system.}
Following the right transition graph in Fig.\ \ref{TRANSGRAPHS}, we suppose the existence of a periodic orbit that returns to its initial state exactly after 6 firings. Assuming w.l.o.g.\ that the trajectory is initially in state $S_{21}$ (see Section 4 for notation), the reset states after firings are then given by the following path
\[
S_{21}\to S_{13}\to S_{13}\to S_{32}\to S_{32}\to S_{32}\to S_{21}\to S_{21},
\]
where the coordinates of the first and last states must coincide in order to comply with (minimal) periodicity. Proceeding similarly as in the proof of Proposition \ref{UNIPER} above, all state coordinates are expressed in terms of the firing times $\{t_i\}_{i=1}^6$. The times themselves are determined by imposing (appropriate) firing at $t=t_i$. At times $t_{2i+1}$ (corresponding to a loop in the same state $S_{k(k-1)}$), the analysis in Section 4.1 implies that the firing conditions write
\[
Wx_2(t_1+0)-\eta=t_2-t_1,\ Wx_1(t_3+0)-\eta=t_4-t_3\ \text{and}\ Wx_3(t_5+0)-\eta=t_6-t_5.
\]
At the other firing times, the condition for the transition $S_{k(k-1)}\to S_{(k-1)(k-2)}$ in the graph yield
\[
t_3-t_2=1-t_2-\tfrac{\eta}{\epsilon},\ t_5-t_4=1-t_4+t_2-\tfrac{\eta}{\epsilon}\ \text{and}\ t_1=1-t_6+t_4-\tfrac{\eta}{\epsilon},
\]
where the periodicity condition $t_7-t_6=t_1$ has been employed in the last equality. 
A system of 6 linear equations results for the $t_i$'s, that has a unique solution such that
\[
t_{2i+3}-t_{2i+1}=t_{2i+2}-t_{2i}=t_2.
\] 
Using this property, one finds that the solution must satisfy $t_1=\frac{(\epsilon-1)\eta}{\epsilon(1-\epsilon)}<0$ which is incompatible with $t_1$ being the first positive firing time; hence the non-existence of the periodic orbit.

\end{document}